\documentclass[a4paper,11pt]{article}
\usepackage{color}
\usepackage{fullpage}
\usepackage{amsmath,amssymb,enumerate}
\usepackage{amsthm}
\usepackage{graphicx}
\usepackage{sidecap}
\sidecaptionvpos{figure}{t}
\usepackage{url}
\usepackage{hyperref}
\usepackage{comment}
\hypersetup{colorlinks=true}
\hypersetup{linkcolor=black,anchorcolor=black,citecolor=black,urlcolor=black}
\hypersetup{pdfauthor= {John Iacono} {Ben Karsin}{Grigorios Koumoutsos}}
\hypersetup{pdfkeywords={point location}{data structures}{dynamic algorithms}{computational geometry}}
\hypersetup{pdftitle=External Memory Planar Point Location with Fast Updates}

\usepackage{cite}

\newtheorem{theorem}{Theorem}[section]

\newtheorem{lemma}[theorem]{Lemma}

\newcounter{note}[section]

\renewcommand{\thenote}{\thesection.\arabic{note}}

\newcommand{\gknote}[1]{\stepcounter{note}\textcolor{blue}{$\ll${\bf Greg~\thenote:} {\sf #1}$\gg$\marginpar{\tiny\bf gk~\thenote}}}

\newcommand{\cI}{\ensuremath{\mathcal{I}}}

\newcommand{\cT}{\ensuremath{\mathcal{T}}}
\newcommand{\Be}{\ensuremath{B^{\epsilon}}}
\newcommand{\Bp}{\ensuremath{B^{\epsilon '}}}
\newcommand{\Bd}{\ensuremath{B^{\epsilon /2}}}
\newcommand{\cL}{\ensuremath{\mathcal{L}}}
\newcommand{\cS}{\ensuremath{\mathcal{S}}}
\newcommand{\cM}{\ensuremath{\mathcal{M}}}

\newcommand{\ios}{I/Os}
\newcommand{\io}{I/O}

\makeatletter
\newcommand*{\rom}[1]{\expandafter\@slowromancap\romannumeral #1@}

\renewcommand{\paragraph}[1]{ \vspace{0.15cm} \noindent \textbf{#1}}

\title{External Memory Planar Point Location with Fast Updates\thanks{This work was supported by the Fonds de la Recherche Scientifique-FNRS under Grant no MISU F 6001 1 and by NSF Grant CCF-1533564.} }

\author{
John Iacono\thanks{Universit\'{e} Libre de Bruxelles.
\texttt{\{johniacono,bkarsin,gregkoumoutsos\}@gmail.com}} \thanks{New York University, USA.}
\qquad
Ben Karsin\footnotemark[2]
\qquad
Grigorios Koumoutsos\footnotemark[2]\\[1.2ex]
Université libre de Bruxelles, Belgium
}

\date{}

\begin{document}

\maketitle

\begin{abstract}
We study dynamic planar point location in the \textit{External Memory Model} or Disk Access Model (DAM). Previous work in this model achieves polylog query and polylog amortized update time. We present a data structure with $O( \log_B^2 N)$ query time and $O(\frac{1}{ B^{1-\epsilon}} \log_B N)$ amortized update time, where $N$ is the number of segments, $B$ the block size and $\epsilon$ is a small positive constant, under the assumption that all faces have constant size. This is a $B^{1-\epsilon}$ factor faster for updates than the fastest previous structure, and brings the cost of insertion and deletion down to subconstant amortized time for reasonable choices of $N$ and $B$. Our structure solves the problem of vertical ray-shooting queries among a dynamic set of interior-disjoint line segments; this is well-known to solve dynamic planar point location for a connected subdivision of the plane with faces of constant size.

\end{abstract}

\setcounter{page}{0}
\newpage

\includecomment{onlymain}
\excludecomment{onlyapp}

\begin{onlymain}
\section{Introduction}

The \textit{dynamic planar point location} problem is one of the most fundamental and extensively studied problems in geometric data structures, and is defined as follows: We are given a connected planar polygonal subdivision $\Pi$ with $N$ edges. For any given query point $p$, the goal is to find the face of $\Pi$ that contains $p$, subject to insertions and deletions of edges. Here we focus on subdivisions $\Pi$ such that each face has constant number of edges. An equivalent formulation, which we use here is as follows: given a set $S$ of $N$ interior-disjoint line segments in the plane, for any given query point $p$, report the first line segment in $S$ that a vertical upwards-facing ray from $p$ intersects, subject to insertions and deletions of segments.

Dynamic planar point location has many applications in spatial databases, geographic information systems (GIS), computer graphics, etc. Moreover it is a natural generalization of the dynamic dictionary problem with predecessor queries; this problem can be seen as the one dimensional variant of planar point location.

In this paper we focus on the External Memory model, also known as the  Disk Access Model (DAM)~\cite{AggarwalV88}.  
The DAM is the standard method of designing algorithms that efficiently execute on large datasets stored in secondary storage.
This model assumes a two-level memory hierarchy, called disk and internal memory and it is parameterized by values $M$ and $B$; the disk is partitioned into blocks of size $B$, of which $M/B$ can be stored in memory at any given moment. The cost of an algorithm in the DAM is the number of block transfers between memory and disk, called Input-Output operations (I/Os). The quintessential DAM-model data structure is the B-Tree \cite{DBLP:journals/acta/BayerM72}. 
See~\cite{Vitter01,Vitter06} for surveys. 
Many applications of dynamic planar point location, such as GIS problems, must efficiently process datasets that are too massive to fit in internal memory, thus it is of great relevance and interest to consider the problem in the DAM and to devise \io\ efficient algorithms.

\subsection{Previous Work}


\paragraph{RAM Model.}
 In the RAM model (the leading model for applications where all data fit in the internal memory) the dynamic planar point location problem has been extensively studied~\cite{ArgeBG06,BaumgartenJM94,ChiangT92,ChengJ92,ChanN15,GoodrichT98}. It is a  major and long-standing open problem in computational geometry to design a data structure that supports queries and updates in $O(\log N)$ time~\cite{Chazelle91,Chazelle94,Snoeyink04}, i.e., to achieve the same bounds as for the dynamic dictionary problem. In a recent breakthrough, Chan and Nekrich in FOCS'15~\cite{ChanN15} presented a data structure supporting queries in $O(\log N (\log \log N)^2)$ time and updates in $O(\log N (\log \log N))$ time. They also showed the tradeoff of supporting queries in $O(\log N)$ time and updates in $O((\log N)^{1+\epsilon})$ time or vice-versa for $\epsilon>0$.

Recently Oh and Ahn~\cite{Gamwtikorea} presented the first data structure for a more general setting where the polygonal subdivision $\Pi$ is not necessarily connected; their data structure supports queries in $O(\log N (\log \log N)^2)$ time and updates in $O(\sqrt{N} \log N (\log \log N)^{3/2} )$ amortized time.

\paragraph{External Memory model} (See Table~\ref{tab:results}). Several data structures have been presented over the years which support queries and updates in polylog($N$) \ios \cite{AgarwalABV99,ArgeV04,ArgeBR12}. Table~\ref{tab:results} contains a list of results of prior work. The best update bound known is by Arge, Brodal and Rao~\cite{ArgeBR12} and achieves $O(\log_B N)$ amortized \ios.  The query time of their data structure is $O(\log^2_B N)$. Very recently, the first data structure that supports queries in $o(\log^2_B N)$ \ios\ was announced by Munro and Nekrich~\cite{MN19}. In particular they support queries in $O(\log_B N (\log \log_B N)^3)$ \ios. However their update time is slightly worse than logarithmic, $O(\log_B N (\log \log_B N)^2)$. In all those works the bounds are obtained by solving the problem of vertical ray-shooting. 


\setlength{\tabcolsep}{0.4em} 
{\renewcommand{\arraystretch}{1.5}

\begin{table}[t]
\begin{center}
\resizebox{\linewidth}{!}{
\begin{tabular}{| c | c | l | l | l | l |}
\hline
Reference & Space & Query Time & Insertion Time & Deletion Time & \\
\hline
Agarwal et al.~\cite{AgarwalABV99} & $O(N)$ & $O(\log^2_{B}{N})$ & $O(\log^2_{B}{N})$ & $O(\log^2_{B}{N})$ & M \\
Arge and Vahrenhold~\cite{ArgeV04} & $O(N)$ &  $O(\log^2_{B}{N})$ & $O(\log^2_{B}{N})$ & $O(\log_{B}{N})$ & G \\
Arge et al.~\cite{ArgeBR12} & $O(N)$ &  $O(\log^2_{B}{N})$ & $O(\log_{B}{N})$ & $O(\log_{B}{N})$ & G \\
Munro and Nekrich~\cite{MN19} & $O(N)$ &  $O(\log_{B}{N}\log^3 \log_B{N})$ & $O(\log_{B}{N}\log^2 \log_B{N})$ & $O(\log_{B}{N}\log^2 \log_B{N})$ & G \\
\hline
This paper & $O(N)$ & $O(\log^2_{B}{N})$ & $O((\log_{B}{N})/B^{1-\epsilon})$ & $O((\log_{B}{N})/B^{1-\epsilon})$ & G \\
\hline
\end{tabular}
}
\end{center}
\caption{Overview of results on dynamic planar point location in external memory. Results marked with M are for monotone subdivisions and G for general ray-shooting among non-intersecting segments. Query bounds are worst-case and update bounds are amortized. Space usage is measured in words. Here $\epsilon $ is a constant such that $ 0 < \epsilon \leq 1/2$. 
}
\label{tab:results}
\end{table}
}

\paragraph{Fast Updates in External Memory.} One of the most intriguing and practically relevant features of the external memory model is that it allows fast updates. For the dynamic dictionary problem with predecessor queries, the optimal update bound in the RAM model is $O(\log N)$. In external memory, however, $B$-trees achieve the optimal query time of $O(\log_B N)$ and typical update time of $O(\log_BN)$, although substantially faster update times are possible. Brodal and Fagerberg~\cite{BrodalF03} showed that $O(\frac{1}{B^{1-\epsilon}}\log_B N)$ amortized \ios\ per update can be supported, for small positive constant, $\epsilon$, while retaining $O(\log_B N)$-time queries; they further showed that this is an asymptotically optimal tradeoff between updates and queries. Observe that this update bound is a huge speedup from $O(\log_B N)$ and that for reasonable choices of parameters, e.g.~$B\geq 1000$, $N<10^{93}$, $\epsilon=\frac{1}{2}$, this yields a subconstant amortized number of \ios\ per update. A similar update bound was later achieved for other dynamic problems like three-sided range reporting and top-$k$ queries~\cite{Brodal16}.

Given this progress and the fact that in the RAM model the bounds achieved for planar point location and the dictionary problem are believed to coincide, it is natural to conjecture that a similar update bound can be achieved for the dynamic planar point location problem. However, to date no result has been presented that achieves sublogarithmic insertion or deletion time.

\subsection{Our Results}

We consider the dynamic planar point location problem in the external memory model and present the first data structure with sublogarithmic amortized update time of $O(\frac{1}{B^{1-\epsilon}}\log_B N)$ \ios. Prior to our work, the best update bound for both insertions and deletions was $O(\log_B N)$, achieved by Arge et al.~\cite{ArgeBR12}.
Our main result is:

\begin{theorem}[Main result]
\label{thm:main}
For any constant $0 < \epsilon \leq 1/2$, there exists a data structure which uses $O(N)$ space, answers planar point location queries for polygonal subdivisions $\Pi$ with faces of constant size in $O((1/\epsilon)^2 \cdot \log^2_B N) = O(\log^2_B N) $ \ios\ and supports insertions and deletions in $O(\log_B N/(\epsilon \cdot B^{1-\epsilon})) = O((\log_B N)/B^{1-\epsilon})$ amortized \ios. The data structure can be constructed in $O((N/B) \log_B N)$ \ios.  
\end{theorem}

To obtain this result, several techniques are used.  
Our primary data structure is
an augmented \textit{interval tree}~\cite{DBLP:journals/ipl/EdelsbrunnerM81}.
We combine both the primary interval tree and two  auxiliary structures described below with the \textit{buffering} technique~\cite{BrodalF03,Arge03} to improve insertion and deletion bounds. In Section~\ref{sec:overview} we prove Theorem~\ref{thm:main} using our auxiliary structures as black boxes and omit some technical details relating to rebuilding; these details are deferred to Section~\ref{sec:restucture}.

Similarly to previous work, we focus on solving the problem of vertical ray-shooting queries. Our first auxiliary structure answers vertical ray-shooting queries among non-intersecting segments whose right (left) endpoints lie on the same vertical line. This is called the \textit{left (right) structure} (in Section~\ref{sec:overview} it will be clear why we choose this terminology and not vice-versa).
Left/Right structures of Agarwal et al.~\cite{AgarwalABV99}, which support queries and updates in $O(\log_B{K})$ \ios, are used by several prior works~\cite{AgarwalABV99,ArgeV04,ArgeBR12}. Our structure improves on their result by reducing the update bound by a factor of $B^{1-\epsilon}$. We obtain the following result, the proof of which is the topic of Section~\ref{sec:left_right}:

\begin{theorem}[Left/right structure]
\label{lem:left_right}
For a set of $K$ non-intersecting segments whose right (left) endpoints lie in the same vertical line and any constant $0 < \epsilon \leq 1/2$, we can create a data structure which supports vertical ray-shooting queries in $ O((1/\epsilon) \cdot \log_B K) = O(\log_B K)$ \ios\ and insertions and deletions in $O((\log_B K)/( \epsilon \cdot B^{1-\epsilon}))=O((\log_B K)/B^{1-\epsilon})$ amortized \ios. This data structure uses $O(K)$ space and it can be constructed in $O((K/B) \log_B K)$ \ios. If the segments are already sorted, it can be constructed in $O(K/B)$ \ios.
\end{theorem}


Our second auxiliary structure answers vertical ray-shooting queries among non-intersecting segments whose endpoints lie in a set of $\Bd+1$ vertical lines. These vertical lines define $\Bd$ vertical slabs, hence the structure is called a \textit{multislab} structure.  We obtain the following result, the proof of which is the topic of Section~\ref{sec:multislab}:

\begin{theorem}[Multislab structure]
\label{lem:multislab}
For any constant $0 < \epsilon \leq 1/2$ and set of $K$ non-intersecting segments whose endpoints lie in $\Bd+1$ vertical lines, we can create a data structure which supports vertical ray-shooting queries in $O( (1/\epsilon) \cdot \log_B K)=O( \log_B K)$ \ios\ and insertions and deletions in $O((\log_B K)/(\epsilon \cdot B^{1-\epsilon}))=O((\log_B K)/B^{1-\epsilon})$ amortized \ios. This data structure uses $O(K)$ space and it can be constructed in $O((K/B) \log_B K)$ \ios. If the segments are already sorted according to a total order, it can be constructed in $O(K/B)$ \ios.
\end{theorem}

A major challenge faced by previous multislab structures is how to efficiently support insertions.  At a high-level, it is hard to deal with insertions in cases where a total order is maintained: each time a new segment gets inserted we need to determine its position in the total order, which cannot be done quickly. Arge and Vitter~\cite{ArgeV04} developed a deletion-only multislab data structure and then used the so-called logarithmic method~\cite{Bentley79} which allowed them to handle insertions in $O(\log_B^2 K)$ \ios. Later Arge, Brodal and Rao~\cite{ArgeBR12} developed a more complicated multislab structure supporting insertions in amortized $O(\log_B K)$ \ios\ by performing separate case analysis depending on the value of $B$. 

Here, we support insertions in a much simpler way by breaking each inserted segment into smaller \textit{unit segments} whose endpoints lie on two consecutive vertical lines and can be compared easily to the segments already stored. This way, we are able to support insertions easily in $O(\log_B K)$ \ios. Finally, we add buffering and obtain sublogarithmic update bounds. 

\subsection{Notation and Preliminaries}

\paragraph{External Memory Model.} Throughout this paper we focus on the external memory model of computation. $N$ denotes the number of segments in the planar subdivision, $B$ the block size and $M$ the number of elements that fit in internal memory. We assume that $M \ll N$ and $2 \leq B \leq \sqrt{M}$ (the \emph{tall cache assumption}). It is well-known that sorting $K$ elements requires $\Theta((K/B) \log_{M/B} (K/B))$ \ios~\cite{AggarwalV88}. Given that $B \leq \sqrt{M} $, this bound is $O((K/B) \log_B K)$. We use this bound for sorting in many places without further explanation. 

\paragraph{{\normalfont{\textbf{Ray-shooting Queries.}}}} In the rest of this paper, we focus on answering vertical ray-shooting queries in a dynamic set of non-intersecting line segments. Let $S$ be the set of segments of the polygonal subdivision $\Pi$. Given a query point $p$, the answer to a vertical ray-shooting query is the the first segment of $S$ hit by a vertical ray emanating from a query point in the $(+y)$ direction. Based on standard techniques (see e.g.~\cite{ArgeV04}), for connected polygonal subdivisions $\Pi$ with faces of size $O(1)$, a planar point location query for a point $p$ can be answered in $O(\log_B N)$ I/Os after answering a vertical ray-shooting query for $p$.

\paragraph{{\normalfont{\textbf{$\Be$-Trees.}}}}
All tree structures that we will use are variants of the
$\Be$-Trees~\cite{BrodalF03} which are $B$-trees except that the internal nodes have at most $\Be$ (and not $B$) children; the leaves still store $\Theta(B)$ data items. For constant $\epsilon$, this does not change the asymptotic height of the tree or the search cost, both remain $O((1/\epsilon) \cdot \log_BN) = O(\log_B N)$. 

\end{onlymain}

\begin{onlymain}
\section{Overall Structure}
\label{sec:overview}

In this Section we prove Theorem~\ref{thm:main}, using the data structures of Theorems~\ref{lem:left_right} and~\ref{lem:multislab} (detailed in Sections~\ref{sec:left_right} and~\ref{sec:multislab}, respectively). 
Given $N$ non-intersecting segments in the plane and a constant $0 < \epsilon \leq 1/2 $, we construct a $O(N)$-space data structure which answers vertical ray-shooting queries in $O( (1/\epsilon)^2 \cdot\log^2_B N) = O(\log^2_B N)$ \ios\ and supports updates in $O((\log_B N)/(\epsilon \cdot B^{1-\epsilon})) =  O((\log_B N)/B^{1-\epsilon})$ amortized \ios. Throughout this section we let $\epsilon ' = \epsilon/2$.

\paragraph{{\normalfont{\textbf{The Data Structure.}}}} As in the previous works on planar point location, our primary data structure is based on the \textit{interval tree} (the external interval tree defined in~\cite{ArgeV03}). Our interval tree $\cI$ is a $\Bp$-tree which stores the $x$-coordinates of segment endpoints in its leaves. Here we assume for clarity of presentation that the interval tree is static, i.e. all new segments inserted share $x$-coordinates with already stored segments; in Section~\ref{sec:restucture} we remove this assumption and extend our data structure to accommodate new $x$-coordinates and achieve the bounds of Theorem~\ref{thm:main}.

Each node of $\cI$ is associated with several secondary structures, as we explain later, and each segment is stored in the secondary structures of exactly one node of $\cI$. Each node $v$ of $\cI$ is associated with a vertical slab $s_v$. The slab of the root is the whole plane. For an internal node $v$, the slab $s_v$ is divided into $\Bp$ vertical slabs $s_1,\dotsc,s_{\Bp}$ corresponding to the children of $v$, separated by vertical lines called \textit{slab boundaries}, such that each slab $s_i$ contains the same number of vertices of $\Pi$ from slab $s_v$.

Let $S$ be the set of segments that compose $\Pi$. Each segment $\sigma \in S$ is \textit{assigned} to a node $v$ of $\cI$. This is the highest node $v$ of $\cI$ such that $\sigma$ is completely contained in slab $s_v$ and intersects at least one slab boundary partitioning $s_v$; if such an internal node $v$ does not exist, then $\sigma$ is assigned to a leaf $v$ such that $\sigma$ is completely contained in its slab $s_v$. Segments assigned to internal nodes are stored in the secondary structures of those nodes, whereas segments assigned to leaves are stored explicitly in the corresponding leaf. By construction of the slab boundaries, each leaf stores $O(B)$ segments in $O(1)$ blocks. 

Consider a segment $\sigma$ assigned to a node $v$ of $\cI$. Let $s_{\ell}$ and $s_r$ be the children slabs of $s_v$ where the left and right endpoints of $\sigma$ lie. We call the segment $\sigma \cap s_{\ell}$ the \textit{left subsegment} of $\sigma$, the segment $\sigma \cap s_r$ the \textit{right subsegment} of $\sigma$ and the rest of $\sigma$ (which spans children slabs $s_{\ell+1},\dotsc,s_{r-1}$) is its middle subsegment. See Figure~\ref{fig:slabs-segments} for an illustration.  In this example, the left subsegment is $\sigma \cap s_5$, the right subsegment is $\sigma \cap s_2$, and the portion of $\sigma$ in $s_3$ and $s_4$ is the middle subsegment. 

\begin{figure}[t]
\centering
\includegraphics[scale= 0.45]{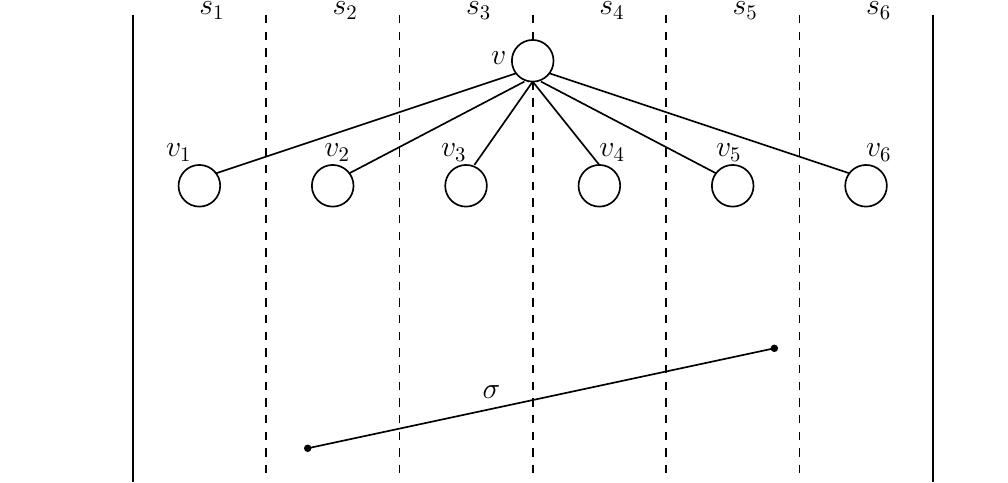}
\caption{The slab of node $v$ of the interval tree $\cI$ is divided into slabs $s_1,\dotsc,s_6$ corresponding to its children $v_1,\dotsc,v_6$. Segment $\sigma$ is assigned to node $v$, with left subsegment in slab $s_2$, right subsegment in $s_5$ and the middle subsegment crosses slabs $s_3,s_4$.}
\label{fig:slabs-segments}
\end{figure}

Let $S_v$ be the set of segments assigned to a node $v$ of $\cI$. To store segments of $S_v$, node $v$ of $\cI$ contains the following secondary structures:

\begin{enumerate}\itemsep.75em 
 \item A multislab structure $\cM$ which stores the set of middle segments.
  \item $\Bp$ left structures $L_i$, for $1 \leq i \leq \Bp $, storing the left (sub)segments of slab $s_i$.
 \item $\Bp$ right structures $R_i$, for $1 \leq i \leq \Bp $, storing the right (sub)segments of slab $s_i$.
 
\end{enumerate}

In addition, each internal node $v$ contains an insertion buffer $I_v$ and deletion buffer $D_v$, each storing up to $B$ segments. 

\paragraph{{\normalfont{\textbf{Construction and Space Usage.}}}} For every node $v$, the buffers $I_v$ and $D_v$ fit in $O(1)$ blocks, since they store at most $B$ segments. By Theorems~\ref{lem:left_right} and~\ref{lem:multislab}, a secondary structure storing $K$ segments uses $O(K)$ space. Since each segment of $S_v$ is stored in at most 3 secondary structures, overall secondary structures of $v$ use $O(|S_v|)$ space. Thus each node $v$ uses $O(|S_v|)$ space. We get that our data structure uses overall $O(\sum_{v \in \cI} |S_v|)=O(N)$ space. The interval tree can be constructed in $O((N/B) \log_B N)$ \ios. This can be done by sorting the segments by their endpoints' $x$-coordinates and then determining all slab boundaries to create a balanced interval tree.
By Theorems~\ref{lem:left_right} and~\ref{lem:multislab}, all secondary structures of a node $v$ of $\cI$ can be constructed in $O((|S_v|/B) \log_B |S_v|)$ \ios\ . Thus, all secondary structures of the tree can be constructed in $O((\sum_{v \in \cI} |S_v|/B) \cdot \log_B N  ) = O((N/B) \log_B N)$ \ios.

\paragraph{{\normalfont{\textbf{Queries.}}}} 
To answer a vertical ray-shooting query for a point $p$, we traverse the root-to-leaf path of $\cI$ based on the $x$-coordinate of $p$, while maintaining a segment $\sigma$ (initialized to \texttt{null}) which is the answer to the query among segments assigned to nodes we have traversed so far. At each node $v$ visited along this path, we first update buffers $I_v$ and $D_v$ by removing from both of them all segments (if any) of $I_v \cap D_v$. Then, we perform a vertical ray-shooting on the secondary structures of $v$; in particular we ray-shoot on the multislab structure and the left and right structures $L_i$ and $R_i$, for $i$ such that the query point $p$ is in slab $s_i$\footnote{Minor detail: For each secondary structure considered, we first perform insertions/deletions of the corresponding segments from buffers $I_v$ and $D_v$.}. After checking the secondary structures, we update $\sigma$ if a closer segment above $p$ is found as a result. Next, we ray-shoot among segments stored in $I_v$ and update $\sigma$ if necessary. Finally, we determine which child $v_i$ of $v$ to visit, and flush any segments of $D_v$ that are contained in the slab of $v_i$ to $D_{v_i}$; this way we make sure that information about deleted segments is updated throughout the root-to-leaf path and no deleted segment can be considered as an answer to the query. We then continue the process at $v_i$. Once a leaf node is reached, we simply compare the $B$ segments it contains with $p$ and return the closest segment above $p$ among them and $\sigma$.

\vspace{0.1cm}
\noindent \textit{Bounding the query cost:} Since any root-to-leaf path of $\cI$ 
has length $O((1/\epsilon')\cdot \log_B N)$, each secondary data structure supports ray-shooting queries in $O((1/\epsilon') \cdot\log_B N)$ \ios\ (due to Theorems~\ref{lem:left_right} and~\ref{lem:multislab}) and we check $O(1)$ secondary structures per node, we
get that a query is answered in $O((1/\epsilon')^2 \cdot \log^2_B N) = O(\log^2_B N)$ \ios. Note that in each node $v$ of the root-to-leaf path visited,  the operations involving $I_v$ and $D_v$ require $O(1)$ \ios, thus they increase the total cost by at most a $O(1)$ factor.

\paragraph{{\normalfont{\textbf{Insertions.}}}} To handle insertions, we use the insertion buffers stored in nodes of $\cI$. When a new segment $\sigma$ is inserted, we insert it in the insertion buffer of the root. Let $v$ be an internal node with children $v_1,\dotsc,v_{\Bp}$. Whenever $I_v$ becomes full, it is flushed. Segments of $I_v$ that cross at least one slab boundary partitioning $s_v$ are inserted in the secondary structures of $v$; segments that are contained in the slab $s_i$ of $v_i$ are inserted in $I_{v_i}$, for $1 \leq i \leq \Bp$. In case $I_v$ becomes full for some node $v$ whose children are leaves, we insert those segments explicitly at the corresponding leaves. When a leaf becomes full, we restructure the tree using split operations on full nodes. 

\vspace{0.1cm}
\noindent \textit{Bounding the insertion cost:}
We compute the amortized cost of an insertion by considering three components:
\begin{enumerate}[(i)]\itemsep.25em
 \item The cost for moving segments between insertion buffers.  Whenever an insertion buffer $I_v$ gets full, it forwards segments to the buffers of its $\Bp$ children performing $O(\Bp)$ \ios. Since a flushing occurs every $B$ insertions in $I_v$, the amortized cost of such operations is $O(\Bp/B) = O(1/(B^{1-\epsilon'}))$.  Each segment will move in at most $O( (1/\epsilon') \log_B N)$ insertion buffers before it is inserted in the secondary structures of a node (or in a leaf). Thus the amortized cost for moving between buffers is $ O((\log_B N)/(\epsilon' \cdot B^{1-\epsilon'}))$.
 
 \item The insertion cost in the secondary structures. By Theorems~\ref{lem:left_right} and~\ref{lem:multislab} we get that insertions in secondary structures require $O((\log_B N)/(\epsilon \cdot B^{1-2\epsilon'}))$ \ios.

 \item The cost of restructuring the tree after insertions when a leaf becomes full. We show in Section~\ref{sec:restucture} that the restructuring requires $O\big(\frac{\log_B N}{\epsilon' \cdot B^{1-\epsilon'}}\big)$ amortized \ios, by slightly modifying our primary interval tree data structure.   
\end{enumerate}
We conclude that our data structure supports insertions in amortized $ O(\log_B N/(\epsilon' \cdot B^{1-2\epsilon'}))  = O(\log_B N/B^{1-\epsilon})$ \ios.

\paragraph{{\normalfont{\textbf{Deletions.}}}} To support deletions, we use the deletion buffers stored in all nodes of $\cI$. To delete a segment $\sigma$, we first check whether $\sigma$ is in the insertion buffer $I_r$ of the root $r$ and in that case we delete it; otherwise we store it in $D_r$. Similar to insertions, whenever $D_v$ gets full for some internal node $v$ with children $v_1,\dotsc,v_{\Bp}$, we flush $D_v$. The segments of $D_v$ crossing at least one slab boundary partitioning $s_v$ are deleted from the corresponding secondary structures associated with $v$; the other segments of $D_v$ are moved to buffers $D_{v_i}$; in case a segment $\sigma$ inserted in $D_{v_i} \cap I_{v_i}$, we delete it from both buffers. In case $D_v$ becomes full for some $v$ parent of leaves, we delete those segments explicitly from the corresponding leaves.

\vspace{0.3cm}
\noindent \textit{Bounding the deletion cost:}
The deletion cost has three components:
\begin{enumerate}[(i)]\itemsep.25em
 \item Moving segments between the deletion buffers. Using the same argument as for insertions, we get 
 that this requires$O(\log_B N/(\epsilon' \cdot B^{1-\epsilon'}))$ \ios, amortized.
 
 \item The cost of deletion in the secondary structures. By Theorems~\ref{lem:left_right} and~\ref{lem:multislab} we get that deletions in secondary structures require amortized $O(\log_B N/(\epsilon' \cdot B^{1-2\epsilon'}))$ \ios.

 \item The cost of restructuring the tree. Every $N/2$ deletions, we rebuild the structure using $O((N/B) \log_B N)$ \ios, to get and amortized restructuring cost of $ O((\log_B N)/B) $ \ios. 
 
\end{enumerate}
Overall deletions are supported in amortized  $O(\log_B N/(\epsilon' \cdot B^{1-2\epsilon'})) = O(\log_B N/(B^{1-\epsilon}))$ \ios.

\end{onlymain}
\begin{onlymain}
\section{Left and Right Structures}
\label{sec:left_right}

In this section we prove Theorem~\ref{lem:left_right}. Given $K$ points all of whose right (left) endpoints lie on a single vertical line, we construct a data structure which answers vertical ray-shooting queries on those segments in $O(\log_B K)$ \ios\ and supports insertions and deletions in $O((\log_B K)/B^{1-\epsilon})$ amortized \ios\, for a constant $0 < \epsilon \leq 1/2$.

We describe the structure for the case where we are given a set $\cL$ of $K$ segments whose right endpoints have the same $x$-coordinate (left structure)\footnote{Recall from Section~\ref{sec:overview} that we call left structures the ones storing the left subsegment of a segment $\sigma$, thus all subsegments stored in a left structure have the same $x$-coordinate of right endpoints.}. The case where the left endpoints of the segments have the same $x$-coordinate (right structure) is completely symmetric. For a segment $\sigma$, we will refer to the $y$-coordinate of its right endpoint as the \textit{$y$-coordinate of $\sigma$}. Conversely we define the \textit{$x$-coordinate of $\sigma$} to be the $x$-coordinate of its left endpoint. 

%

\paragraph{Total Order.} We assume that the segments in $\cL$ are ordered according to their $y$-coordinates. We can always order the segments according to this total order in $O((K/B) \log_{B}K)$ \ios. 

\paragraph{The Data Structure.} We store all segments of $\cL$ in an augmented $\Be$-tree $\cT$ which supports vertical ray-shooting queries, insertions and deletions. The degree of each node is between $\Be/2$ and $\Be$, except the root which might have degree in the range $[2,\Be]$, and leaves store $\Theta(B)$ elements. For a node $v \in \cT$, let $\cT_v$ be the subtree rooted at $v$. Since the segments are sorted according to their $y$-coordinates, each subtree $\cT_v$ corresponds to a range of $y$-coordinates, which we call the $y$-range of node $v$. Let $v$ be an internal node of $\cT$ with children $v_1,\dotsc,v_{\Be}$. Node $v$ stores the following information:

\begin{enumerate}\itemsep.65em
\item A buffer of segments $\cS_v$ of capacity $B$ which contains segments in the $y$-range of $v$ whose left endpoints have the smallest $x$-coordinates (i.e., segments that extend the farthest from the vertical line) and are not stored in any buffer $\cS_w$ for an ancestor $w$ of $v$. In other words, $\cT$ together with segments of buffers $\cS_v$ form an external memory priority search tree~\cite{ArgeSV99}.
\item An insertion buffer $I_v$ and a deletion buffer $D_v$, each storing up to $B$ segments.
\item A list $\cM_v$ that contains, for each child $v_i$, the segment with minimum $x$-coordinate stored in $\cS_{v_i}$. We call this the \textit{minimal} segment for child $v_i$. 
\end{enumerate}

The data structure satisfies the following invariants: For each node $v \in \cT$, either $|\cS_{v}| \geq B/2$ or if $|\cS_{v}| < B/2$, then $I_v$ and $D_v$ are empty and all buffers stored in descendants $v$ are empty. Also, for each node $v$, buffers $\cS_v, I_v$ and $D_v$ are disjoint. Finally, for a leaf $v$, $I_v$ and $D_v$ are empty.

\paragraph{Construction and Space Usage.}  Overall buffers and lists of each node contain $O(B)$ segments, i.e. they can be stored in $O(1)$ blocks. Thus $\cT$ can be stored in $O(K/B)$ blocks, i.e. it requires $O(K)$ space. Construction of $\cT$ requires $O(\frac{K}{B} \log_B K)$ \ios, since we need to sort all $K$ segments according to their $y$-coordinates. If the segments are already sorted according to their $y$-coordinate, then $\cT$ can be created in $O(K/B)$ \ios.

\paragraph{Queries in the static structure.}
To get a feel for how our structure supports queries, we first show how to perform queries in the static case, i.e., assuming there are no insertions and deletions and all buffers $I_v$ and $D_v$ are empty. Later we will give a precise description of performing queries in the fully dynamic structure.

Let $\rho^{+}$ be the ray emanating from $p$ in the $(+y)$ direction and $\rho^{-}$ the ray emanating from $p$ in the $(-y)$ direction. We query the structure by finding the first segment hit by both $ \rho^+$ and $\rho^{-}$. We keep two pointers, $v_+$ and $v_-$, initialized at the root. We also keep the closest segments $\sigma_+$ and $\sigma_-$ seen so far in the $(+y)$ and $(-y)$ direction respectively (initialized to $+ \infty$ and $- \infty$). At each step, we update both $v_+$ and $v_-$ to move from a node of depth $i$ to a node of depth $i+1$. While at level $i$, $v_-$ and $v_+$ might coincide, or one of them might be undefined (set to \texttt{null}).

\end{onlymain}

\begin{onlymain}

We now describe the query algorithm.  We start at the root of $\cT$ and advance down, while updating $v_+$, $v_-$, and $\sigma_+$,$\sigma_-$. When at depth $i$, we find the first segment $\sigma_i$ hit by $\rho^{+}$ among $\cS_{v_-}$ and $\cS_{v_+}$ and update $\sigma_+$ if necessary (i.e. if $\sigma_i$ is the first segment hit by $\rho^{+}$ among all segments seen so far). Similarly, we ray-shoot on $\rho^{-}$ among $\cS_{v_-}$ and $\cS_{v_+}$ and update $\sigma_-$ if necessary. To determine in which nodes of depth $i+1$ to continue the search, we ray-shoot on $\rho^+$ among $\cM_{v_-}$ and 
$ \cM_{v_+} $ and also ray-shoot on $\rho^{-}$ among $\cM_{v_-}$ and 
$ \cM_{v_+} $ (i.e., all minimal segments of children of $v_-$ and $v_+$). Let $\sigma_{m+}$ be the first segment in $M_{v+} \cup M_{v-}$  hit by $\rho^{+}$ (if such a segment exists) and $v_s$ be the node containing $\sigma_{m+}$  (if $\sigma_{m+}$ exists). If the $y$-range of $v_s$ is higher than the $y$-coordinate of $\sigma_+$ or if $\sigma_{m+}$ does not exist, we leave $v_+$ undefined for level $i+1$. Otherwise, we set $v_+ = v_s $. Similarly, call $\sigma_{m-}$ the first minimal segment of $M_{v+} \cup M_{v-}$ hit by $\rho^{-}$ and $v_p$ be the node containing $\sigma_{m-}$ (if such a segment exists). If the $y$-range of $v_p$ is lower than the $y$-coordinate of $\sigma_-$ or if $\sigma_{m-}$ does not exist, we leave $v_-$ undefined for level $i+1$. Otherwise we set $v_- = v_p$. 

If both $v_+$ and $v_-$ are undefined for the next level $i+1$, we stop the procedure and output $\sigma_+$ as the result to the vertical ray-shooting query. Otherwise we repeat the same procedure in the next level. When we reach a leaf level, we find the first segment hit by $\rho^{+}$ among $\cS_{v_-}$ and $\cS_{v_+}$, update $\sigma+$ if necessary, and output $\sigma+$ as the result of the query. 

\begin{figure}[t]
\centering
\includegraphics[scale= 0.45]{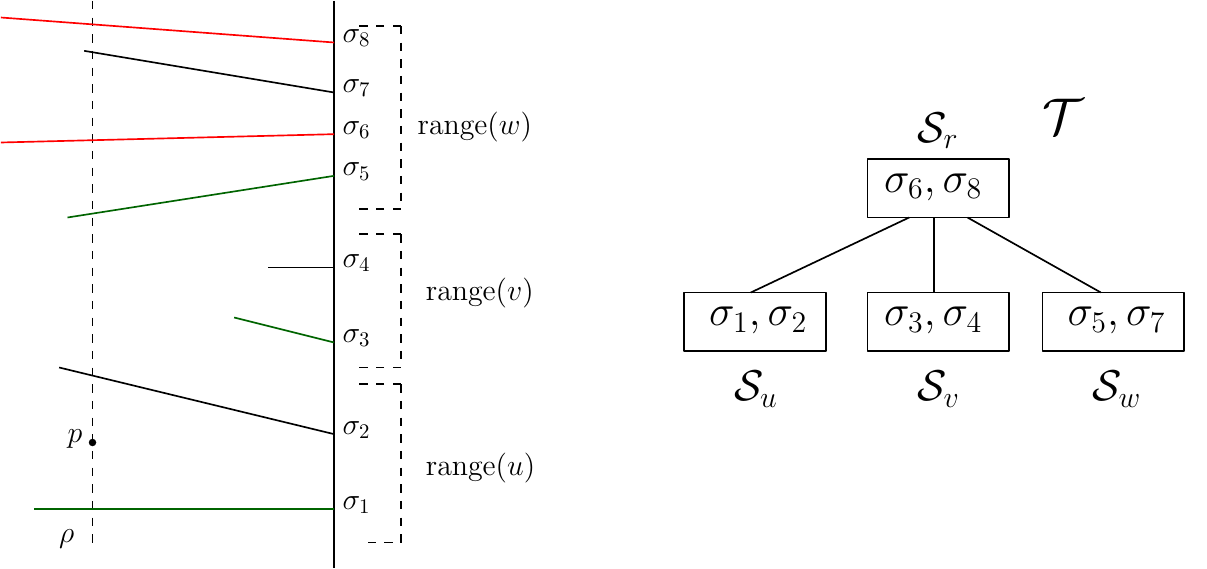}
\caption{Example of the query algorithm in the left structure: Left column shows the segments stored in $\cT$, the query point $p$ and the vertical ray $\rho$ emanating from $p$. Right column shows buffers $\cS$ of the nodes of $\cT$. Red segments are stored in the root. For nodes $u,v,w$, the green segment is their minimal segment, i.e., the one stored in list $\cM_r$. By ray-shooting on $\rho$ among green segments, the first segment hit upwards is $\sigma_5$, which is stored in $\cS_w$, thus we set $v_+ = w$. Note that $\sigma_2$ (the correct answer for the query) is not stored in $\cS_w$, i.e., maintaining only $v_+$ produces an incorrect answer. Thus, our algorithm ray-shoots downwards as well, hitting $\sigma_1$, which is stored in $u$, and setting $v_- = u$. Then, by ray-shooting on $\rho$ among $\cS_u$ and $\cS_w$, the first segment we hit upwards of $p$ is $\sigma_2$.}
\label{fig:left_right_query}
\end{figure}

\vspace{0.1cm}
\textit{Remark:} The reader might wonder why we answer vertical ray-shooting queries in both directions and keep two pointers $v_-$ and $v_+$. Isn't it sufficient to answer queries in one direction and keep one pointer at each step? Figure~\ref{fig:left_right_query} shows an example where this is not true and maintaining only the $v_+$ pointer would result in an incorrect answer. 

The formal proof of correctness of this query algorithm is deferred to Appendix~\ref{sec:left_right_app}. 

\end{onlymain}

\begin{onlyapp}
\section{Queries in the Left and Right Structures.}
~\label{sec:left_right_app}

In this Section we give further details on the left (right) structure which were omitted from Section~\ref{sec:left_right}.

\paragraph{Queries.} We begin with the queries and we show the correctness of the query algorithm of the static left (right) structure.

\vspace{0.1cm}

\textit{Correctness:}
The correctness of the query algorithm follows from the next lemma. For a node $v \in \cT$ let $S_{v}$ be the set of segments stored in buffers $\cS$ in $\cT_v$ .

\begin{lemma}
\label{lem:left_right_query_general}
Assume that at the end of the $i$th step of the query algorithm, either $v_+$ or $v_-$ is defined. Then $\sigma_+$ is the first segment hit by $\rho^+$ among the segments of $\cL - (S_{v_-} \cup S_{v_+})$.
\end{lemma}

\begin{proof}
We prove the lemma by induction. 

\vspace{0.1cm}

\textit{Induction Base:} At the end of the first step, $v_+$ and $v_-$ are children of the root $r$ and $\sigma_+$ is the first segment hit by $\rho^+$ among all segments stored at the root (in $\cS_r$ and $\cM_r$). By definition of $v_s = v_+$, for any child of the root $v$ with higher $y$-range than $v_+$, $\sigma_+$ is below all segments of $S_{v}$. Similarly, for any child of the root $v'$ with smaller $y$-range than $v_-$ (if $v_-$ exists), there is no segment in $S_{v'}$ hit by $\rho^+$ (since there exists a segment in $\cS_{v_-}$ hit by $\rho^-$). Finally, for any child $v''$ of the root whose $y$-range is between the range of $v_-$ and $v_+$, by definition of $v+$, there is no segment in $S_{v''}$ hit by $\rho^{+}$. We conclude that $\sigma_+$ is the first segment hit by $\rho^+$ among the segments in $\cL - (S_{v_-} \cup S_{v_+})$.

\vspace{0.1cm}

\textit{Inductive Step:} Assume the lemma holds at the end of step $i$, i.e. we have at least one of $v_+$ and $v_-$ at level $i$ and $\sigma_+$ is the first segment hit by $\rho^+$ among all segments in $\cL - (S_{v_+} \cup S_{v_+} )$. 

During $(i+1)$th step we ray-shoot on $\rho^{+}$ among segments stored in $\cS_{v_+}, \cS_{v_-}, \cM_{v_+}$ and $\cM_{v_-}$, and update $\sigma_+$ if necessary. Let $v_s$ be the node containing the first segment hit by $\rho^+$ among $\cM_{v_+}$ and $\cM_{v_-}$ (if such a segment exists). Let also $v_p$ be the node containing the first segment hit by $\rho^-$ among $\cM_{v_+}$ and $\cM_{v_-}$ (if such a segment exists).

 By definition of $v_s $, for any node $v$ which is a child of $v_-$ or $v_+$ with higher $y$-range than $v_s$, $\sigma_+$ is below all segments of $S_{v}$. Similarly, for a node $v'$ which is a child of $v_-$ or $v_+$ with smaller $y$-range than $v_p$ (if $v_p$ exists), there is no segment of $S_{v'}$ hit by $\rho^+$ (since there exists a segment in $\cS_{v_p}$ hit by $\rho^-$). Finally, for any child $v''$ of $v_-$ or $v_+$ whose $y$-range is between the range of $v_-$ and $v_+$, by definition of $v+$, there is no segment in $S_{v''}$ hit by $\rho^{+}$. 
 
 Recall that by the induction hypothesis $\sigma_+$ at the end of the previous step was the first segment hit by $\rho^+$ among segments of $ \cL - (S_{v_+} \cup S_{v_+})$. Now we updated $\sigma_+$ and showed that there is no segment hit by $\rho^{+}$ before $\sigma_+$ in any subtree other than $\cT_{v_s}$ or $\cT_{v_p}$. We conclude that $\sigma$ is the first segment hit by $\rho^+$ among the segments in $\cL - (S_{v_s} \cup S_{v_p})$. Since at the end of the $(i+1)$th step we set $v_- = v_p$ and $v_+ = v_s$, the lemma follows. 

\end{proof}

We now explain how Lemma~\ref{lem:left_right_query_general} implies the correctness of the query algorithm. To see that, let $i$ be the last level where either $v_+$ or $v_-$ is defined; at the beginning of the query algorithm at level $i$, $\sigma_+$ is the first segment hit by $\rho^{+}$ among segments of $\cL - (S_{v_-} \cup S_{v_+})$. Moreover at the end of this step, both $v_s$ and $v_p$ are not defined, i.e., for each child $v$ of $v_-$ or $v_+$ there is no segment in $S_v$ hit by $\rho^{+}$ before $\sigma_+$. Since $S_{v_-} \cup S_{v_+} = \cS_{v_-} \cup \cS_{v_+} \cup (\cup_{v} S_v) $, we get that $\sigma_+$ is the first segment hit by $\rho^{+}$ among segments of $\cL - (\cS_{v_-} \cup \cS_{v_+})$. By checking all segments of $\cS_{v_-} \cup \cS_{v_+}$ and updating $\sigma_+$ if necessary, we make sure that $\sigma_+$ is the first segment hit by $\rho^{+}$ among segments of $\cL$. 

\end{onlyapp}

\begin{onlymain}
\vspace{0.1cm}

\textit{Bounding the query cost:} To count the cost, observe that in each step we move down the tree by one level and perform operations that require $O(1)$ \ios, as we check $O(B)$ segments stored in the current nodes $v_-$ and $v_+$. Since the height of the tree is $O((1/\epsilon) \log_B K)$, a query is answered in $ O((1/\epsilon) \log_B K)) = O(\log_B K)$ \ios.



\paragraph{Insertions.}
Assume we want to insert a segment $\sigma$ into the left structure $\cL$. If the $x$-value of $\sigma$ is smaller than the maximum $x$-value of a segment stored in the buffer of the root $\cS_r$, we insert $\sigma$ into $\cS_r$. Otherwise we store $\sigma$ in the insertion buffer of the root $I_r $. Note that insertion of $\sigma$ in $\cS_r$ might cause $\cS_r$ to overflow (i.e., $|\cS_r| = B+1 $); in that case we move the segment of $\cS_r$ with the maximum $x$-value into the insertion buffer of the root $I_{r}$.

 Let $v$ be an internal node with children $v_1,\dotsc,v_{\Be}$. Whenever the insertion buffer $I_v$ becomes full, we \textit{flush} it, moving the segments to buffers of the corresponding children. For a segment $\sigma$ that should be stored in child $v_i$, we repeat the same procedure as in the root: Check whether $\sigma$ has smaller $x$-value than the maximum $x$-value of a segment stored in $\cS_{v_i}$ and if yes, store $\sigma$ in $\cS_{v_i}$, otherwise store it in $I_{v_i}$. If $\cS_{v_i}$ overflows, we move its last segment (i.e. the one with maximum $x$-value) into $I_{v_i}$. Also, if $\sigma$ gets stored in $\cS_{v_i}$ and its $x$-value is smaller than all previous segments of $\cS_{v_i}$, we update the minimal segment of $v_i$, $\cM_{v}$.  
 
 When $\cS_v$ overflows for some leaf $v$, we split $v$ into two leaves $v_1$ and $v_2$, as in standard $B$-trees. Note that this might cause recursive splits of nodes at greater height.  

\vspace{0.15cm}
\noindent \textit{Bounding the insertion cost:} To flush a buffer $I_v$ and forward segments to buffers $\cS_{v_i}$ and $I_{v_i}$, for $1\leq i \leq \Be$ we perform $O(\Be)$ \ios. Since $I_v$ becomes full after at least $B$ insertions, the amortized cost of moving a segment from $I_v$ to buffers of a child of $v$ is $O(\Be/B) = O(1/B^{1-\epsilon})$. Each inserted segment moves between buffers in a root-to-leaf path of length $O((1/\epsilon) \log_B K)$, thus the total amortized cost for moves between buffers is $O(\log_B K/(\epsilon \cdot B^{1-\epsilon}))$ \ios. The restructuring of $\cT$ due to splitting nodes requires amortized $O(1/B)$ \ios, as in standard  B-trees. Thus, insertions are supported in $ O(\log_B K/(\epsilon \cdot B^{1-\epsilon})) $ amortized \ios.

\paragraph{Deletions.} To delete a segment $\sigma$, we first check whether it is stored in the buffers of the root $\cS_r$ or $I_r$; in this case we delete it. Otherwise, we insert $\sigma$ in the deletion buffer of the root $D_r$. 

Let $v$ be an internal node with children $v_1,\dotsc,v_{\Be}$. Whenever $D_v$ becomes full we flush it and move the segments to the corresponding children and repeat the same procedure: For a segment $\sigma$ which moves to child $v_i$, we check whether it is stored in $\cS_{v_i}$ or $I_{v_i}$: if yes, we delete it and update the minimal segment of $v_i$ in $\cM_v$ if necessary. Otherwise, we store $\sigma$ in the deletion buffer $D_{v_i}$. If segment buffer $\cS_{v}$ underflows (i.e., $| \cS_{v} | < B/2 $), we refill it using segments stored in buffers $\cS_{v_i}$; the segments moved to $\cS_{v}$ are deleted from $\cS_{v_i}$ and all necessary updates in $\cM_{v}$ are performed. This might cause underflowing segment buffers $\cS_{v_i}$ for children of $v_i$; we handle those in the same way. In case all buffers $\cS_{v_i}$ become empty and $| \cS_{v} | < B$ , we move the segments from $I_v$ to $\cS_v$ until either $| \cS_{v} | = B$ or $| I_v | = 0$.

\vspace{0.1cm}
\noindent \textit{Bounding the deletion cost:} Deletion cost consists of three components:

\begin{enumerate}[(i)]\itemsep.25em
\item Cost for moving segments between buffers: Using the same analysis as for insertions we get that this requires $O(\log_B K/(\epsilon \cdot B^{1-\epsilon})) $ amortized \ios.
\item Cost due to refilling of buffers $\cS_{v}$: For a node $v$ with children $v_i$, while refilling buffer $\cS_{v}$ from $\cS_{v_i}$ we perform $O(\Be)$ \ios\ and we move $\Theta(B)$ segments one level higher. Thus the amortized cost of moving a segment up by one level is $O(1/B^{1-\epsilon})$. Since the tree has height $O((1/\epsilon) \cdot \log_B K)$, over a sequence of $K$ deletions the total number of moves of segments by one level is $O((1/\epsilon) \cdot K \cdot \log_B K)$. Thus the total cost due to refilling is at most $O( (1/ \epsilon B^{1-\epsilon}) K \cdot \log_B K)$, which implies that the amortized cost is $O(\log_B K/(\epsilon \cdot B^{1-\epsilon}))$.

A corner case that we did not take into account above is when the total number of segments stored in buffers $\cS_{v_i}$ are less than $B/2$. In this case it is not valid that the amortized cost of updating $\cS_{v}$ is $O(\Be/B)$. To take care of this, we use a simple amortization trick: we double charge all \ios\ performed relating to insertions.  This way, for each buffer $\cS_{v_i}$ there is a saved \io\ from the time when segments move from $I_v$ to node $v_i$. We use this additional saved \io\ when $\cS_{v_i}$ gets emptied due to the refilling of $S_{v}$.

\item Restructuring requires $O(\frac{\log_B K}{B})$ amortized \ios, by rebuilding the structure after $K/2$ deletions.  
\end{enumerate}

Overall, the amortized deletion cost is $O(\log_B K/(\epsilon \cdot B^{1-\epsilon})) =O(\log_B K/B^{1-\epsilon})$ \ios. 



\paragraph{Queries in the dynamic structure.}
We now describe how to extend our query algorithm to the dynamic case.
In order to ensure that all nodes visited are up-to-date and we do not miss any updates in the insertion/deletion buffers, when moving a pointer from a node $u$ to its child $v_i$, we flush any deletes in $D_{u}$ to $v_i$, i.e. delete segments of $D_u$ that are stored in $\cS_{v_i}$, store the other segments in $D_{v_i}$ and update $\cM_u$ if necessary. We then delete any segments found in both $I_{v_i}$ and $D_{v_i}$. Finally, we compare segments in $I_{v_i}$ with $\sigma_+$ (recall this is the first segment hit by $\rho^{+}$ among segments considered so far) and, if any segment in $I_{v_i}$ would be hit by $\rho^{+}$ before $\sigma_+$ we replace $\sigma_+$ with it. Clearly this increases the total cost by at most a $O(1)$ factor compared to the static case, thus the query cost is $O((1/\epsilon)\log_B K)$ \ios.
\end{onlymain}

\begin{onlymain}
\section{Multislab Structure}
\label{sec:multislab}

In this section we prove Theorem~\ref{lem:multislab}. Assume that we are given a set of $K$ non-intersecting segments with endpoints on at most $\Bd+1$ vertical lines $l_1,\dotsc,l_{\Bd+1}$, for some constant $O < \epsilon \leq 1/2$. We show that those segments can be stored in a data structure which uses $O(K)$ space, supports vertical ray-shooting queries in $ O(\log_B K)$ \ios, and updates in $ O(\log_B K/B^{1-\epsilon})$ amortized \ios, for $0 < \epsilon \leq 1/2$. This data structure can be constructed in $O((K/B) \log_{B}K)$ \ios. We call this data structure a \textit{multislab} structure.

For notational convenience we set $\epsilon' = \epsilon/2$. This way endpoints of the segments lie on at most $\Bp +1$ vertical lines $l_1,\dotsc,l_{\Bp+1}$.  For $1 \leq i \leq \Bp$, let $s_i$ denote the vertical slab defined by vertical lines $l_i$ and $l_{i+1}$.
We will show that queries are supported in $O(\log_B K)$ \ios\ and updates in $O((\log_B K) / B^{1-2\epsilon '})$ \ios . Theorem~\ref{lem:multislab} then follows.

\paragraph{Total Order.} In order to implement the multislab structure we need to maintain an ordering of the segments based on their $y$-coordinates. Using standard approaches (see e.g.~\cite{ArgeV04,ArgeBR12}) we can define a partial order for segments that can be intersected by a vertical line. Arge et. al.~\cite{ArgeVV07} showed how to extend a partial order into a total order on $K$ segments (not necessarily all intersecting the same vertical line) in $O((K/B) \log_{M/B}\frac{K}{B}) = O((K/B) \log_{B}K)$ \ios. We use this total order to create our multislab structure.

\paragraph{The Data Structure.} 
We store the ordered segments in an augmented B-tree \cT which supports queries, insertions and deletions. The degree of each node is between $\Bp /2$ and $\Bp$, except the root which might have degree in the range $[2,\Bp]$. Leaves store $\Theta(B)$ elements. For a node $v \in T$, let $\cT_v$ be the subtree rooted at $v$. Let $v_1,\dotsc,v_{\Bp}$ be the children of an internal node $v$. Node $v$ stores the following information:

\begin{enumerate}\itemsep.2em
 \item A buffer $\cS_v$ of capacity $B$ which contains the highest (according to the total order) segments stored in $\cT_v$ which are not stored in any buffer $\cS_w$ for an ancestor $w$ of $v$. In other words, $\cT$ together with segments of buffers $\cS_v$ form an external memory priority search tree~\cite{ArgeSV99}. 
 \item An insertion buffer $I_v$ and a deletion buffer $D_v$, both storing up to $B$ segments.

 \item A list $L_v$ which contains, for each slab $s_i$, $1 \leq i\leq \Bp$, and each child $v_j$, $1 \leq j\leq \Bp$, the highest segment (according to the total order) $t_{i,j}$ crossing slab $s_i$ stored in $\cT_{v_j}$. 
\end{enumerate}

The data structure satisfies the following invariants: i) for each node $v \in \cT$, either $|\cS_{v}| \geq B/2$ or if $|\cS_{v}| < B/2$, then $I_v$ and $D_v$ are empty and all buffers of descendants $w$ of $v$ are empty, ii) for each node $v$, buffers $\cS_v, I_v$ and $D_v$ are disjoint, and iii) for every leaf $v$, $I_v$ and $D_v$ are empty. 

\paragraph{Construction and Space Usage.}  Overall buffers of each node contain $O(B)$ segments and list $L_v$ contains at most $B^{2\epsilon'} = O(B)$ segments, i.e., they can be stored in $O(1)$ blocks. Thus $\cT$ can be stored in $O(K/B)$ blocks, i.e. it requires $O(K)$ space. The structure can be constructed in $O(\frac{K}{B} \log_B K)$ \ios. If segments are already sorted according to a total order, construction requires $O(K/B)$ \ios.

\paragraph{Insertions.} To insert a new segment $\sigma$ we need to determine its position in the total order. Clearly, we can not afford to produce a new total order from scratch, as this costs $O((K/B) \log_{B}K)$ \ios. Thus, we break $\sigma$ into at most $\Bp$ \textit{unit segments}, where each segment crosses exactly one slab. In particular, if $\sigma$ crosses slabs $s_{\ell},\dotsc,s_r$, we break it into unit segments $\sigma_{\ell},\dotsc,\sigma_r$, where segment $\sigma_i$ crosses slab $s_i$. We call all such unit segments stored in $\cT$ \textit{new} segments. The rest of the segments stored in $\cT$ are called the \textit{old} segments of $\cT$. Now we can easily update the total order: segment $\sigma_i$ needs to be compared only with segments crossing slab $s_i$; if $\sigma_p$ and $\sigma_s$ are the predecessor and successor of $\sigma_i$ within slab $s_i$, we locate $\sigma_i$ in an arbitrary position between $\sigma_p$ and $\sigma_s$ in the total order. This way a valid total order is always maintained. 

We now describe the insertion algorithm. When segment $\sigma $ needs to be inserted, we first break it into unit segments $\sigma_{\ell},\dotsc,\sigma_r$. For each segment $\sigma_j$, $\ell \leq j \leq r$, we first check whether it should be inserted in the buffer $\cS_r$ of the root: if this is the case we store it there; otherwise we store it in the insertion buffer of the root $I_r$. In case $\cS_{r} $ overflows (i.e. $|\cS_{r}| = B+1$) we move its last segment (according to the total order) to $I_{r}$. Let $v$ be an internal node with children $v_1,\dotsc,v_{\Bp}$. Each time $I_v$ becomes full, we \textit{flush} it and move the segments to its children $v_i$, for $1 \leq i \leq \Bp$. For a segment moving from $v$ to $v_i$, we first check whether it is greater (according to the total order) than the minimum segment stored in $\cS_{v_i}$ and if so we store it in $\cS_{v_i}$; otherwise we store it in buffer $I_{v_i}$. In case $\cS_{v_i} $ overflows (i.e. $|\cS_{v_i}| = B+1$) we move its last segment to $I_{v_i}$. Also we update information in list $L_v$ if necessary. In case $I_{v_i}$ becomes full, we repeat the same procedure recursively.

When $\cS_v$ overflows for some leaf $v$, we split $v$ into two leaves $v_1$ and $v_2$, as in standard $B$-trees. Note that this might cause recursive splits of nodes at greater height.  

\vspace{0.15cm}

\noindent \textit{Bounding the insertion cost:} To flush a buffer $I_v$ and move segments to buffers of child nodes $\cS_{v_i}$ and $I_{v_i}$, we need to perform $O(\Bp)$ \ios. Since each segment breaks into at most $\Bp$ unit segments, a buffer of size $B$ becomes full after at least $B/ \Bp = B^{1-\epsilon'}$ insertions. Thus the amortized cost of moving a segment from a buffer of depth $i$ to depth $i+1$ is $O(\Bp/{B^{1-\epsilon'}}) = O(1/B^{1-2\epsilon'})$. Since each segment will be eventually stored in a node of depth $O((1/\epsilon') \cdot \log_B K)$, the amortized cost until it gets inserted is $O(\log_B K/(\epsilon ' \cdot B^{1-2\epsilon'}))$.  The restructuring of $\cT$ due to splitting full nodes requires amortized $O(1)$ \ios, as in standard B-trees. Overall insertions require $O(\log_B K/(\epsilon \cdot B^{1-2\epsilon'}) ) = O((\log_B K)/B^{1-\epsilon} ) $ amortized \ios.

\vspace{0.15cm}

\noindent \textit{Linear space usage:} To avoid increases in space usage due to unit segments, whenever there are $K/\Bp$ new segments, we rebuild the structure. This way the space used is $O(K + (K/\Bp) \cdot \Bp) = O(K)$. This rebuilding requires $O((K/B)\log_B K)$ \ios, i.e., $ O(\log_B K/B^{1-\epsilon'})$ amortized \ios, thus it does not violate the insertion time bound. 

\paragraph{Deletions.} 
The process of deleting a segment, $\sigma$, is similar to insertion: we break $\sigma$ into at most $\Bp$ unit segments $\sigma_{\ell},\dotsc,\sigma_r$ where $s_\ell$ and $s_r$ are the leftmost and rightmost slabs spanned by $\sigma$ and apply the deletion procedure for each of those unit segments separately.

The deletion algorithm for a unit segment $\sigma_i$ is analogous to the one of the left (right) structure of Section~\ref{sec:left_right}. For completeness we describe it here. To delete a unit segment $\sigma_i$, we first check whether it is stored in the buffers of the root $\cS_r$ or $I_r$; in this case we delete it. Otherwise, we insert $\sigma_i$ in the deletion buffer of the root $D_r$. Let $v$ be an internal node with children $v_1,\dotsc,v_{\Bp}$. Whenever $D_v$ becomes full we flush it and forward the segments to the corresponding children and repeat the same procedure: For a segment $\sigma$ which moves to child $v_i$, we check whether it is stored in $\cS_{v_i}$ or $I_{v_i}$ and if this is the case, we delete it and update list $L_v$ if necessary. Otherwise, we store $\sigma_i$ in the deletion buffer $D_{v_i}$.  

In case segment buffer $\cS_{v}$ underflows (i.e., $| \cS_{v} | < B/2 $), we refill it using segments from buffers $\cS_{v_i}$; segments moved to $\cS_{v}$ are deleted from $\cS_{v_i}$ and $L_{v}$ gets updated (if needed). This might cause underflowing segment buffers $\cS_{v_i}$; we handle those in the same way. In case all buffers $\cS_{v_i}$ become empty and $| \cS_{v} | < B$ , we move to $\cS_{v}$ the segments from $I_v$ until either $| \cS_{v} | = B$ or $| I_v | = 0$. After $K/\Bp$ deletions we rebuild our data structure.

\vspace{0.15cm}
\textit{Remark:} Note that here we split all segments $\sigma$ into unit segments $\sigma_{\ell},\dotsc,\sigma_r$. However, the old segments $\sigma$ are not unit segments and are stored manually in the data structure. However this does not affect our algorithm: whenever the first unit segment $\sigma_i$ which is a part of $\sigma$ reaches the node $v$ such that $\sigma \in \cS_{v}$, we delete $\sigma$ from $\cS_v$ and remove $\sigma_i$ from deletion buffers. The remaining segments $\sigma_j$ will eventually reach node $v$ and realize that $\sigma$ is already deleted from $\cS_v$; at this point $\sigma_j$ gets deleted.

\vspace{0.15cm}
\noindent \textit{Bounding the deletion cost:} The analysis of the deletion cost is identical to the analysis of deletions in the structure of Section~\ref{sec:left_right}. Since each segment breaks into at most $\Bp$ unit segments, we get an amortized deletion cost of $O(\log_B K/ B^{1-2\epsilon'}) = O(\log_B K/ B^{1-\epsilon})$.

\vspace{0.15cm}

\noindent \textit{Linear space usage:} Similar to insertions, we need to make sure that the total space used is not increasing asymptotically due to the use of at most $\Bp$ unit segments in deletion buffers for each deleted segment $\sigma$. The total capacity of deletion buffers is $O(K)$. Since we rebuild the structure after $K/\Bp$ deletions, there are at most $O(K)$ segments stored in deletion buffers, i.e., deletion buffers never get totally full and total space used is $O(K)$

\begin{figure}[t]
\centering
\includegraphics[scale= 0.65]{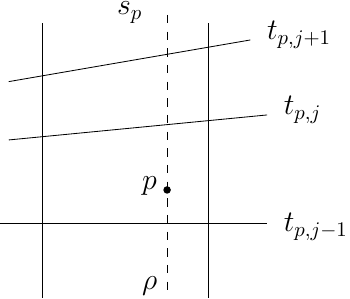}
\caption{Vertical ray-shooting queries in the multislab structure: Query point $p$ is in slab $s_p$. $\rho$ is the vertical ray emanating from $p$. While being at node $v$ of $\cT$, to decide in which child to continue our search we examine all minimal segments $t_{p,1}, \dotsc,t_{p,\Bp}$ stored in list $L_v$. Among them, the first one hit by $\rho$ is $t_{p,j}$. Thus the search continues at child $v_j$ of $v$.}
\label{fig:multislab_query}
\end{figure}

\paragraph{Queries.} 
Let $p$ be the query point and $\rho^{+}$ be the the vertical ray emanating from $p$ in the $(+y)$ direction. Let also $s_p$ be the slab containing $p$. We can find $s_p$ in $O(1)$ \ios\ by storing all slab boundaries in a block. We perform a root-to-leaf search and we keep the first segment $\sigma$ hit by $\rho^{+}$ among segments seen so far. While visiting a node $v$ we do the following: (i) perform a vertical ray-shooting query from $p$ among segments stored in buffers $\cS_{v}$ and $I_{v}$, and update $\sigma$ if necessary (ii) move to the child $v_i$ which contains the successor segment $t_{p,j}$ of $p$ in list $L_v$ (see Figure~\ref{fig:multislab_query}) and (iii) find in $I_v$ (resp. $D_v$) the segments crossing slab $s_p$ and should be stored (according to the total order) in $\cT_{v_i}$ and move them to $\cS_{v_i}$ or $I_{v_i}$ (resp. delete them from $\cS_{v_i}$ or store it in $D_{v_i}$). If a segment inserted in $D_{v_i}$ is also stored in $I_{v_i}$, we delete it from both buffers.
Once we reach a leaf $v$, we first delete from $\cS_v$ the segments that are in the deletion buffer of its parent and then we perform ray-shooting query among the segments stored in $\cS_v$ and update $\sigma$ if necessary.

\textit{Bounding the query cost:} Since we follow a root-to-leaf path, and at each level we need to perform $O(1)$ \ios, a ray-shooting query is answered in $O((1/\epsilon') \cdot \log_B K) $ \ios.


\end{onlymain}

\begin{onlymain}
\section{Counting the Restructuring Cost}
\label{sec:restucture}

In Section~\ref{sec:overview} we proved the Theorem~\ref{thm:main} (query and update bounds of the overall structure) without taking into account the cost of restructuring the interval tree $\cI$ due to insertions that cause leaves to become full. In this section we show that Theorem~\ref{thm:main} holds while taking into account the restructuring of $\cI$ as well.

When a leaf becomes full we need to split it. This split in turn might cause the split of the parent and possibly continue up the tree, thus causing some part of the tree $\cI$ to need rebalancing. While rebalancing, we need to perform updates in the secondary structures so that they are adjusted with the updated nodes of the interval tree $\cI$. In this section, we show that we can slightly modify our data structure such that all updates in secondary structures can be performed in $O(\frac{\log_B N}{B^{1-\epsilon}})$ amortized \ios. This implies that Theorem~\ref{thm:main} holds. 

\paragraph{Our Approach.} We use a variant of the weight-balanced \Be-tree of~\cite{ArgeV03}. Each leaf stores at most $B$ segment endpoints. Let $v$ be a node at height $h-1$ with parent $p(v)$. Node $p(v)$ stores $w_v = \Theta(B \cdot B^{\epsilon h})$ elements in its subtree $\cI_{p(v)}$. We will show that if node $v$ splits, then we can perform all updates needed in the secondary structures in $O(w_v/B^{1-\epsilon})$ \ios. This implies that a split requires amortized $O(1/B^{1-\epsilon})$ \ios, since after a restructuring, there should be at least $\Omega(w_v)$ insertions in $\cI_{p(v)}$ until the next split is needed. Since each insertion can cause $O(\log_B N)$ splits, we get an amortized restructuring cost of $O(\frac{\log_B N}{B^{1-\epsilon}})$ \ios\ for insertion.

\paragraph{Splitting a node.} Node $v$ splits into two new nodes $v_1$ and $v_2$. The slab $s_v$ of $v$ is divided into two slabs $s_{v_1},s_{v_2}$ with slab boundary $b$; see Figure~\ref{fig:slabs-split}. To capture this change and update our data structure, we need to perform updates in the secondary structures of $p(v)$ and construct the secondary structures for $v_1,v_2$. We describe these updates in detail and show that they can be performed in $O(w_v/B^{1-\epsilon})$ \ios. In our analysis we use the fact that all secondary structures (multislab and left/right) storing $K$ segments can be scanned in $O(K/B) $ \ios.

\begin{figure}[t]
\centering
\includegraphics[scale= 0.6]{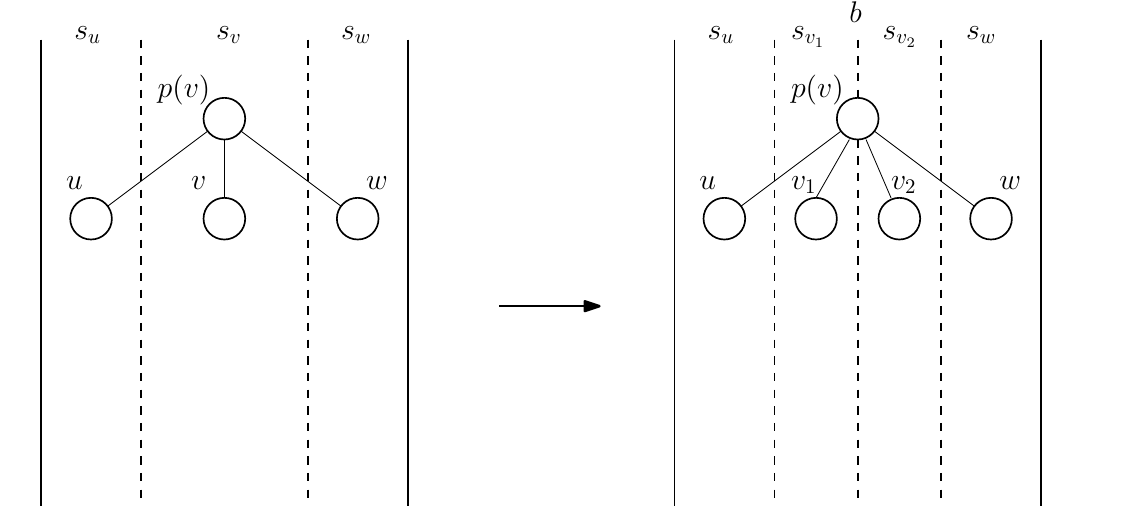}
\caption{Splitting a node $v$ into $v_1$ and $v_2$: slab $s_v$ is divided into slabs $s_{v_1}$ and $s_{v_2}$ with boundary $b$.}
\label{fig:slabs-split}
\end{figure}

 \paragraph{Updates in secondary structures of $p(v)$.}
 We begin with the construction of left/right structures for $v_1$ and $v_2$ using the previous left/right structures for $v$.  We describe the creation of left structures $L_{v_1}$ and $L_{v_2}$ for $v_1$ and $v_2$, respectively, and the right structures are symmetric. Segments that were stored in $L_v$ and do not cross $b$ (like segment $\sigma_1$ in Figure~\ref{fig:segments-examples}) are stored in $L_{v_2}$; segments of $L_v$ that cross $b$ (see segment $\sigma_2$ in Figure~\ref{fig:segments-examples}) are stored in $L_{v_1}$. To identify if a segment is stored in $L_{v_1}$ or $L_{v_2}$ we just need to scan $L_v$, which takes $O(w_v/B)$ \ios. Moreover, there are some additional segments that need to be stored in left/right structures of $p(v)$: the segments that are strictly inside the slab of $v$ (i.e. they were stored in secondary structures of $v$) and cross $b$; see e.g. segment $\sigma_3$ in Figure~\ref{fig:segments-examples}. For those segments, their left subsegments are stored in $L_{v_1}$ and their right subsegments in $R_{v_2}$. To find such segments we need to scan all secondary structures stored at $v$. Since each secondary structure can be scanned in $O(w_v/B)$ \ios\ and there are $O(\Be)$ structures stored in each node, all this takes $O((w_v/B) \cdot \Be) = O(w_v/B^{1-\epsilon})$ \ios.

 We now proceed to the updates of the multislab structure of $p(v)$. Here, we just need to add some segments to the previous multislab structure. The new segments are the segments of $L_v$ that cross $b$ which are not already stored in the multislab (and symmetrically, the segments of $R_v$ that cross $b$ and are not yet in the multislab). For an example, see segment $\sigma_2$ in Figure~\ref{fig:segments-examples}; before it was not stored in the multislab and now we store its middle subsegment. Note that the middle subsegment is a unit segment (i.e. crosses exactly one slab) thus we don't need to compute a new total order; we can find its position in the total order by comparing it only with segments that cross slab $s_{v_2}$. All those segments that need to be added can be found by scanning $L_v$ and $R_v$ in $O(w_v/B)$ \ios. Insertions in the multislab of $p(v)$ require $O((\log_B w_v)/B^{1-\epsilon}) = O(w_v/B)$ \ios. Also, all information stored in nodes of the multislab structure can be updated in $O(w_v/B)$ \ios. Overall, all updates in the  multislab structure of $v$ are performed in $O(w_v/B)$ \ios.

\begin{figure}[t]
\centering
\includegraphics[scale= 0.6]{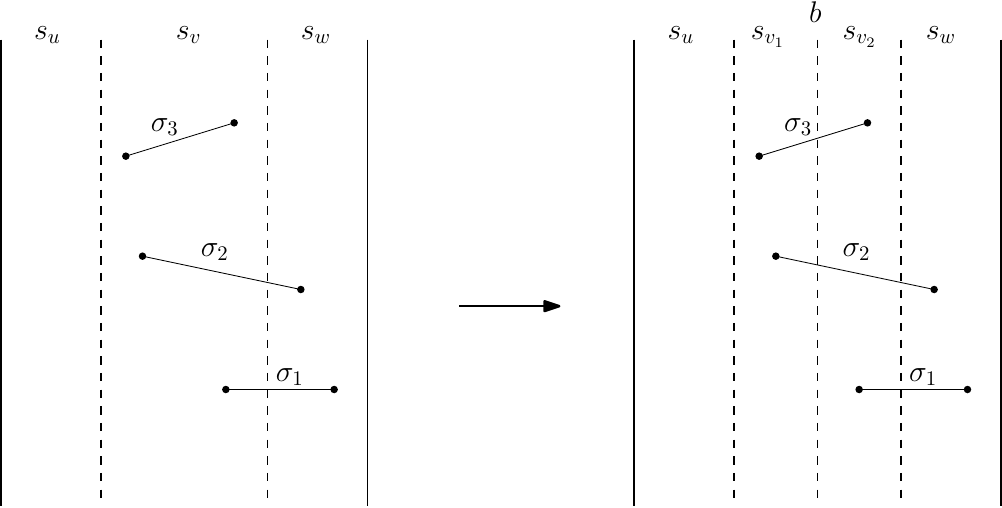}
\caption{Example of segments that get stored in different secondary structures after a split.  Segment $\sigma_1$ was stored in $L_v$ and, after the split, gets stored in $L_{v_2}$. Segment $\sigma_2$ was stored in $L_v$; following the split its left subsegment is stored in $L_{v_1}$ and its middle subsegment in the multislab structure of $p(v)$. Segment $\sigma_3$ was previously stored in secondary structures of $v$, and after the split it should be stored in structures $L_{v_1}$ and $R_{v_2}$ of $p(v)$. }
\label{fig:segments-examples}
\end{figure}

 \paragraph{Construct secondary structures for $v_1$ and $v_2$.} 
 The left and right structures for each child slab of $v_1$ and $v_2$ will be based on the left/right structure of the same slab in $v$ just by removing the segments that cross $b$ (which are assigned to $p(v)$ as we explained above).  Similarly, segments that cross $b$ are excluded from the multislab structure. 
 
  We start with the construction of left/right structures of $v_1$ and $v_2$. We describe the left and the right is symmetric. For each slab $s_k$ of $v$, $1\leq k \leq \Be$ we scan the left list $L_k$; the segments that do not cross $b$ remain in $L_k$ and the others are deleted. All this takes $O((w_v/B) \cdot \Be) = O(w_v/B^{1-\epsilon}) $ \ios.

 Finally we create the multislab structures for $v_1$ and $v_2$. Again, we need to scan the multislab of $v$ and delete the segments that cross $b$, which takes $O(w_v/B)$ \ios. Then we need to build the multislabs of $v_1$ and $v_2$ out of the remaining segments. Since all segments are already sorted according to a total order, this can be done in $O(w_v/B) $ \ios.

\end{onlymain}
\section{Concluding Remarks}
\label{sec:openprobs}

\begin{onlymain}
We presented the first data structure with sublogarithmic update time for dynamic planar point location in the DAM, matching the update bound achieved by $\Be$-trees for the dictionary problem. Moreover, until the very recent work of Munro and Nerich~\cite{MN19} in SOCG'19, our query bound $O(\log^2_B N)$ was the best known for the problem.  Since in~\cite{MN19} authors achieved the first $o(\log^2_B N)$ query bound, a very interesting research direction is to achieve the ``best of both worlds", i.e. describing a data structure with the query bound of~\cite{MN19} and the update time of the data structure presented in this work. 
We conjecture that the optimal bounds for dynamic planar point location in external memory are $O(\log_B N)$ for queries and $O(\log_B N/B^{1-\epsilon})$ (the bound we achieved in this work) for updates. 

\end{onlymain}


\bibliographystyle{plain}
{\small \bibliography{point_location_external} }

\newpage
\excludecomment{onlymain}
\includecomment{onlyapp}
\appendix

\end{document}